\documentclass[12pt]{article}
\usepackage{a4wide}
\usepackage{amssymb}
\begin{document}
{\renewcommand{\thefootnote}{\fnsymbol{footnote}}
\begin{center}
{\LARGE  Properties of fluctuating states\\ in loop quantum cosmology}\\
\vspace{1.5em}
Martin Bojowald\footnote{e-mail address: {\tt bojowald@gravity.psu.edu}}
\\
\vspace{0.5em}
Institute for Gravitation and the Cosmos,\\
The Pennsylvania State
University,\\
104 Davey Lab, University Park, PA 16802, USA\\
\vspace{1.5em}
\end{center}
}

\setcounter{footnote}{0}

\newtheorem{theo}{Theorem}
\newtheorem{prop}{Proposition}
\newtheorem{lemma}{Lemma}
\newtheorem{defi}{Definition}

\newcommand{\proofend}{\raisebox{1.3mm}{\fbox{\begin{minipage}[b][0cm][b]{0cm}
\end{minipage}}}}
\newenvironment{proof}{\noindent{\it Proof:} }{\mbox{}\hfill \proofend\\\mbox{}}
\newenvironment{ex}{\noindent{\it Example:} }{\medskip}
\newenvironment{rem}{\noindent{\it Remark:} }{\medskip}

\newcommand{\case}[2]{{\textstyle \frac{#1}{#2}}}
\newcommand{\lP}{\ell_{\mathrm P}}

\newcommand{\md}{{\mathrm{d}}}
\newcommand{\tr}{\mathop{\mathrm{tr}}}
\newcommand{\sgn}{\mathop{\mathrm{sgn}}}

\newcommand*{\R}{{\mathbb R}}
\newcommand*{\N}{{\mathbb N}}
\newcommand*{\Z}{{\mathbb Z}}
\newcommand*{\Q}{{\mathbb Q}}
\newcommand*{\C}{{\mathbb C}}

\begin{abstract}
 In loop quantum cosmology, the values of volume fluctuations and correlations
 determine whether the dynamics of an evolving state exhibits a bounce. Of
 particular interest are states that are supported only on either the positive
 or the negative part of the spectrum of the Hamiltonian that generates this
 evolution. It is shown here that the restricted support on the spectrum does
 not significantly limit the possible values of volume fluctuations. 
\end{abstract}

\section{Introduction}

A solvable model \cite{BouncePert} that captures basic features of classical
and quantum cosmology is given by two canonical variables, $Q$ and $P$ with
Poisson bracket $\{Q,P\}=1$, and a 1-parameter family of Hamiltonians,
$H_{\delta}=|Q\sin(\delta P)|/\delta$ with $\delta\geq 0$. In the limit
$\delta\to0$, $H_0=|QP|$ is quadratic up to the absolute value, and a system
close to an upside-down harmonic oscillator is obtained. Since $QP$ and
therefore ${\rm sgn}(QP)$ is preserved by equations of motion generated by the
auxiliary Hamiltonian $H_0'=QP$, the set of regular solutions (such that
$QP\not=0$) of the classical $H_0$-system is given by the union of two
disjoint sets: solutions of the $H_0'$-system with initial values $Q(0)$,
$P(0)$ such that ${\rm sgn}(Q(0)P(0))=1$, and solutions of the $-H_0'$-system
with initial values $Q(0)$, $P(0)$ such that ${\rm sgn}(Q(0)P(0))=-1$. All
classical solutions can therefore be obtained from a quadratic Hamiltonian.

For $\delta\not=0$, $H_{\delta}=|{\rm Im}J_{\delta}|/\delta$ is, up to the
absolute value, linear in $J_{\delta}:= Q\exp(i\delta P)$, whose real and
imaginary parts, together with $Q$, are generators of the ${\rm sl}(2,{\mathbb
  R})$ algebra
\begin{equation}
 \{Q,{\rm Re}J_{\delta}\}= -\delta {\rm Im}J_{\delta}\quad,\quad \{Q,{\rm
   Im}J_{\delta}\}= \delta {\rm Re}J_{\delta}\quad,\quad \{{\rm
   Re}J_{\delta},{\rm Im}J_{\delta}\}= \delta Q\,.
\end{equation}
Again, introducing an auxiliary Hamiltonian $H_{\delta}'={\rm
  Im}J_{\delta}/\delta$, all regular solutions (such that ${\rm
  Im}J_{\delta}\not=0$) of the $H_{\delta}$-system can be obtained from
solutions of the $\pm H_{\delta}'$-systems with suitable initial values.

In a simple cosmological interpretation, $|Q|$ is proportional to the volume
of an expanding or collapsing universe, while $P$ is proportional to the
Hubble parameter. According to the Friedmann equation of classical cosmology
for flat spatial slices, $H_0=|QP|$ can be interpreted as the momentum
canonically conjugate to a free, massless scalar source $\phi$, whose energy
density $\rho\propto\frac{1}{2}p_{\phi}^2/Q^2$ with the momentum $p_{\phi}$
canonically conjugate to $\phi$ is then required to be proportional to
$P^2$. Solutions $Q(\phi)$ and $P(\phi)$ of Hamilton's equations of motion
generated by $\pm H_0\propto p_{\phi}$ therefore describe how $Q$ and $P$
change in relation to the ``internal time'' $\phi$. If $\delta\not=0$,
$H_{\delta}$ can still be interpreted in this way, but only if the Friedmann
equation is modified such that $P^2$ is replaced by $\sin^2(\delta
P)/\delta^2$. This modification may be motivated by the appearence of
holonomies in loop quantum gravity \cite{LoopRep,ALMMT} and loop quantum
cosmology \cite{IsoCosmo}, and is supposed to describe an implication of
quantum geometry.

Replacing the unbounded function $P^2$ with a bounded function $\sin^2(\delta
P)/\delta^2$, still proportional to the energy density of a matter source,
suggests that the classical big-bang singularity, at which the energy density
diverges, could be avoided by quantum-geometry effects
\cite{GenericBounce}. Indeed, solutions for $Q(\phi)$ of equations of motion
generated by $\pm H_{\delta}'$,
\begin{equation}
 \frac{{\rm d}Q}{{\rm d}\phi}= \pm {\rm Re}J_{\delta}(\phi)\quad,\quad
 \frac{{\rm d}{\rm Re} J_{\delta}}{{\rm d}\phi}=\pm Q(\phi)\,,
\end{equation}
are superpositions of real exponential functions. If the condition
$Q^2-|J_{\delta}|^2=0$ is imposed, which ensures that $P$ in the definition of
$J_{\delta}$ is real, the equation
\begin{equation} \label{QJ}
 Q^2-({\rm Re}J_{\delta})^2= ({\rm Im}J_{\delta})^2= (\delta H_{\delta}')^2>0\,,
\end{equation}
which is by definition positive for regular solutions, implies that $Q(\phi)$
must be cosh-like and ${\rm Re}J_{\delta}(\phi)$ sinh-like. The eternally
collapsing behavior of the volume $Q(\phi)$ approaching zero if $\delta=0$,
$Q(\phi)=Q(0)\exp(\pm\phi)$, is then replaced by a ``bounce'' at the non-zero
minimum of cosh.

The preceding argument ignores quantum fluctuations, which may be expected to
be significant in a discussion of big-bang solutions. If $(\Delta Q)^2$ is
large, it could conceivable change the balance of signs in (\ref{QJ}), in
which $\langle\hat{Q}^2\rangle=\langle\hat{Q}\rangle^2+(\Delta Q)^2$ would
take the place of $Q^2$. For states with $(\Delta Q)^2\geq (\delta
H_{\delta}')^2+(\Delta{\rm Re}J_{\delta})^2$, the right-hand side of
(\ref{QJ}), written for expectation values, is no longer positive, and
$\langle\hat{Q}\rangle(\phi)$ would not be cosh-like. The possibility of such
non-bouncing solutions in loop quantum cosmology has been demonstrated using
canonical effective methods \cite{NonBouncing}, in particular for small
$\delta H_{\delta}'$ relevant for an understanding of generic spacelike
singularities \cite{Infrared,EFTLQC}.

However, for quantum states the absolute value in $H_{\delta}$ has to be
treated with greater care than in the case of classical solutions. Solutions
of quantum evolution generated by an operator $\hat{H}_{\delta}$ via a
Schr\"odinger equation for wave functions can be expressed as superpositions
of solutions of quantum evolution generated by an operator
$\hat{H}_{\delta}'$, provided the latter are supported solely on the positive
or negative part of the spectrum of $\hat{H}_{\delta}'$. (See
Sec.~\ref{s:Exist} below for a demonstration.) This condition is a
straightforward replacement of the classical restriction on initial
values. But it may have more significant ramifications, in particular when
quantum fluctuations are taken into account that may be larger than the
expectation value $\langle\hat{H}_{\delta}\rangle$, as required to change the
signs in (\ref{QJ}). A state that is supported only on the positive part of
the spectrum of $\hat{H}_{\delta}'$ and has an expectation value of
$|\hat{H}_{\delta}'|$ close to zero may not have arbitrarily large fluctuations
of $\hat{H}_{\delta}'$. The question to be addressed in this paper is whether
this restriction also limits the size of fluctuations of $\hat{Q}$.

\section{Eigenstates}

We will first determine the spectra of $\hat{H}_0'$ and $\hat{H}_{\delta}'$
and then discuss relevant properties of states obtained from superpositions of
their positive parts.

\subsection{Eigenstates of $\hat{H}_0'$}

We use the symmetric ordering
\begin{equation}
 \hat{H}_0'= \frac{1}{2}(\hat{Q}\hat{P}+\hat{P}\hat{Q})
\end{equation}
to quantize $H_0'=QP$ on the standard $L^2$-Hilbert space. Eigenstates of this
operator in the $Q$ and $P$-representations are determined by the same type of
first-order differential equation,
\begin{equation}
 Q\frac{{\rm d}\psi_{\lambda}(Q)}{{\rm
     d}Q}+\frac{1}{2}\psi_{\lambda}(Q)=i\frac{\lambda}{\hbar}\psi_{\lambda}(Q)
\end{equation}
in the $Q$-representation, and
\begin{equation}
 P\frac{{\rm d}\phi_{\lambda}(P)}{{\rm
     d}P}+\frac{1}{2}\phi_{\lambda}(P)=-i\frac{\lambda}{\hbar}\phi_{\lambda}(P)
\end{equation}
in the $P$-representation. For every $\lambda$, there are in each
representation two orthogonal solutions $\psi_{\lambda\pm}(Q)$ and
$\phi_{\lambda\pm}(P)$, respectively, given by
\begin{eqnarray}
 \psi_{\lambda+}(Q)&=&\left\{\begin{array}{cl} 0 & \mbox{if }Q\leq 0\\
     c_{\lambda+} Q^{i\lambda/\hbar-1/2} & \mbox{if
     }Q>0\end{array}\right.\\
 \psi_{\lambda-}(Q)&=&\left\{\begin{array}{cl} c_{\lambda-}  
(-Q)^{i\lambda/\hbar-1/2} & \mbox{if }Q<0\\     0 & \mbox{if 
     }Q\geq 0\end{array}\right.\\
\phi_{\lambda+}(P)&=&\left\{\begin{array}{cl} 0 & \mbox{if }P\leq 0\\
     d_{\lambda+} P^{-i\lambda/\hbar-1/2} & \mbox{if
     }P>0\end{array}\right.\\
 \phi_{\lambda-}(P)&=&\left\{\begin{array}{cl} d_{\lambda-}  
(-P)^{-i\lambda/\hbar-1/2} & \mbox{if }P<0\\     0 & \mbox{if 
     }P\geq0\end{array}\right. \,.
\end{eqnarray}

It is obvious that $\psi_{\lambda_1+}$ and $\psi_{\lambda_2-}$ are orthogonal
to each other for any $\lambda_1$ and $\lambda_2$, and so are
$\phi_{\lambda_1+}$ and $\phi_{\lambda_2-}$. Moreover, 
\begin{eqnarray}
 \int_{-\infty}^{\infty} \psi_{\lambda_1\pm}^*(Q) \psi_{\lambda_2\pm}(Q){\rm
   d}Q&=& c_{\lambda_1\pm}c_{\lambda_2\pm} \int_0^{\infty}
 q^{i(\lambda_2-\lambda_1^*)/\hbar} \frac{{\rm 
     d}q}{q}\nonumber\\
&=&c_{\lambda_1\pm}c_{\lambda_2\pm}
 \int_{-\infty}^{\infty} \exp(ix(\lambda_2-\lambda_1^*)/\hbar) {\rm
   d}x\nonumber\\
&=&
 \left\{\begin{array}{cl} 2\pi
 \hbar c_{\lambda_1\pm}c_{\lambda_2\pm} \delta(\lambda_2-\lambda_1^*) &\mbox{if
   }\lambda_2-\lambda_1^*\in{\mathbb
     R}\\\infty&\mbox{otherwise}\end{array}\right.
\end{eqnarray}
and
\begin{eqnarray}
 \int_{-\infty}^{\infty} \phi_{\lambda_1\pm}^*(P) \phi_{\lambda_2\pm}(P){\rm
   d}P&=&d_{\lambda_1\pm}d_{\lambda_2\pm} \int_0^{\infty}
 p^{-i(\lambda_2-\lambda_1^*)/\hbar} \frac{{\rm 
     d}p}{p}\nonumber\\
&=& d_{\lambda_1\pm}d_{\lambda_2\pm}
 \int_{-\infty}^{\infty} \exp(-iy(\lambda_2-\lambda_1^*)/\hbar) {\rm
   d}y\nonumber\\
&=& 
\left\{\begin{array}{cl} 2\pi
 \hbar d_{\lambda_1\pm}d_{\lambda_2\pm} \delta(\lambda_2-\lambda_1^*) &\mbox{if
   }\lambda_2-\lambda_1^*\in{\mathbb R}\\\infty
   &\mbox{otherwise}\end{array}\right. 
\end{eqnarray}
where the substitutions $q=|Q|$, $x=\log q$, $p=|P|$ and $y=\log p$ have been
used. For real $\lambda$, all eigenstates are delta-function normalizable,
fixing the coefficients
$c_{\lambda\pm}=1/\sqrt{2\pi\hbar}=d_{\lambda\pm}$. The spectrum of
$\hat{H}_0'$ is therefore real, continuous, and twofold degenerate.

\subsection{Eigenstates of $\hat{H}_{\delta}'$}

For $\delta\not=0$, the Hamiltonian is periodic in $P$ with period
$2\pi/\delta$. To be specific, we will assume that the basic operators are
represented on a separable Hilbert space of square-integrable functions
periodic in $P$, such that $\hat{Q}$ has a discrete spectrum given by
$\hbar\delta{\mathbb Z}$. Inequivalent representations, such as states which
are periodic only up to a phase factor $\exp(i\epsilon)$, for which the
spectrum of $\hat{Q}$ is shifted by $\hbar\epsilon$, or non-separable Hilbert
spaces as used often in loop quantum cosmology \cite{Bohr}, would not change
our results. In the $Q$-representation, our states therefore obey an $\ell^2$
inner product such that
\begin{equation}\label{l2}
(\psi_1,\psi_2)=\sum_{n=-\infty}^{\infty}
\psi_1(n\hbar\delta)^* \psi_2(n\hbar\delta)\,.
\end{equation}

We write $\hat{H}_{\delta}'$ as
\begin{equation} \label{HOp}
 \hat{H}_{\delta}'= \frac{{\rm Im}\hat{J}_{\delta}}{\delta}=
 \frac{1}{2i\delta}\left(\hat{J}_{\delta}-\hat{J}_{\delta}^{\dagger}\right)=
 \frac{1}{2i\delta} \left(\hat{Q}\exp(i\delta\hat{P})-
   \exp(-i\delta\hat{P})\hat{Q}\right)\,.
\end{equation}
Since $\exp(i\delta\hat{P})$ is a translation operator in the
$Q$-representation, eigenstates of $\hat{H}_{\delta}'$ in this representation
are determined by a difference equation
\begin{equation}\label{Diff}
 (Q+\hbar\delta)\psi_{\lambda}(Q+\hbar\delta)+2i\delta\lambda\psi_{\lambda}(Q)- Q
 \psi_{\lambda}(Q-\hbar\delta)=0
\end{equation}
where $Q$ takes the values $n\hbar\delta$ with integer $n$.  This equation
with non-constant coefficients does not have straightforward solutions.  It
is, however, possible to show that eigenstates obey a similar twofold
degeneracy as in the case of $\hat{H}_0'$:
\begin{lemma} \label{l:Symm}
  For given $\lambda$, there are two orthogonal solutions $\psi_{\lambda\pm}$,
  one of which is supported on positive values of $Q$ (and $Q=0$), and one on
  negative values of $Q$. They are related by
\begin{equation} \label{psipm}
 \psi_{\lambda\pm}(Q)=\psi_{\lambda\mp}(-Q-\hbar\delta)\,.
\end{equation}
\end{lemma}
\begin{proof}
  Let us first look for solutions such that
  $\psi_{\lambda+}(-\hbar\delta)=0$. Using the equation (\ref{Diff}) for
  $Q=-\hbar\delta$, we obtain $\psi_{\lambda+}(-2\hbar\delta)=
  2i(\lambda/\hbar)\psi_{\lambda+}(-\hbar\delta)=0$. Moreover, if
  $\psi_{\lambda+}(-(n-1)\hbar\delta)=0$ and
  $\psi_{\lambda+}(-n\hbar\delta)=0$, using the equation for $Q=-n\hbar\delta$
  shows that $\psi_{\lambda+}(-(n+1)\hbar\delta)=0$. By induction,
  $\psi_{\lambda+}(Q)=0$ for all integer $Q/\hbar\delta<0$. However, if
  $\psi_{\lambda+}(0)\not=0$ for such a solution,
  $\psi_{\lambda+}(\hbar\delta)=-2i(\lambda/\hbar)\psi_{\lambda+}(0)\not=0$,
  using (\ref{Diff}) for $Q=0$. The solution therefore is not identically
  zero, and it is unique up to multiplication with a constant
  $\psi_{\lambda+}(0)$.

  A similar line of arguments, starting with the assumption that
  $\psi_{\lambda-}(0)=0$, implies that $\psi_{\lambda-}(Q)=0$ for all
  integer $Q/\hbar\delta\geq 0$, while assuming
  $\psi_{\lambda-}(-\hbar\delta)\not=0$ guarantees that the solution is
  not identically zero. Since the supports of any $\psi_{\lambda_1+}$ and
  $\psi_{\lambda_2-}$ are disjoint, the two states are orthogonal with respect
  to the inner product (\ref{l2}).

  Substituting $-Q-\hbar\delta$ for $Q$ in (\ref{Diff}), we obtain the
  equation
\begin{eqnarray}
 0&=&-Q\psi_{\lambda}(-Q)+2i\delta\lambda\psi_{\lambda}(-Q-\hbar\delta)+
 (Q+\hbar\delta)  \psi_{\lambda}(-Q-2\hbar\delta)=0\nonumber\\
&=&-Q\bar{\psi}(Q-\hbar\delta)+2i\delta \lambda\bar{\psi}(Q)+
(Q+\hbar\delta)\bar{\psi}(Q+\hbar\delta)
\end{eqnarray}
equivalent to (\ref{Diff}).  The definition
$\bar{\psi}(Q)=\psi(-Q-\hbar\delta)$ introduced in the second line maps a
function $\psi$ supported on non-negative integers (times $\hbar\delta)$ to a
function $\bar{\psi}$ supported on negative integers (times $\hbar\delta$),
and vice verse. Applied to solutions of (\ref{Diff}), it therefore maps
$\psi_{\lambda\pm}$ to $\psi_{\lambda\mp}$.
\end{proof}

In the $P$-representation, eigenstates of (\ref{HOp}) obey the first-order
differential equation
\begin{equation}
 \sin(\delta P)\frac{{\rm d}\psi_{\lambda\pm}(P)}{{\rm d}P}+ \frac{1}{2}\delta
 \exp(i\delta P)\psi_{\lambda\pm}(P)=
 -i\frac{\lambda\delta}{\hbar}\psi_{\lambda\pm}(P)\,. 
\end{equation}
This equation is solved by
\begin{eqnarray}
 \psi_{\lambda+}(P) &=& \left\{\begin{array}{cl} 0 & \mbox{if }\pi\leq\delta
     P\leq 2\pi
     \\\displaystyle \sqrt{\frac{\delta}{2\pi\hbar}} \frac{(\cot(\delta
       P/2))^{i\lambda/\hbar}}{\sqrt{\sin(\delta P)}} \exp(-i\delta P/2) &
     \mbox{if }0<\delta P<\pi\end{array}\right.\\
 \psi_{\lambda-}(P) &=& \left\{\begin{array}{cl}\displaystyle
     \sqrt{\frac{\delta}{2\pi\hbar}} 
     \frac{(-\cot(\delta 
       P/2))^{i\lambda/\hbar}}{\sqrt{-\sin(\delta P)}} \exp(-i\delta P/2) &
     \mbox{if }\pi<\delta P<2\pi\\ 0 &
     \mbox{if }0\leq\delta P\leq \pi\end{array}\right.\,.
\end{eqnarray}
The substitution $x=\log |\cot(\delta P/2)|$ shows that these states are
delta-function normalized. The spectrum therefore has the same properties as
in the case of $\hat{H}_0'$, being real, continuous, and twofold degenerate.

\subsection{Existence of positive-energy solutions with large fluctuations} 
\label{s:Exist}

For any $\delta$, completeness of the eigenstates of a self-adjoint operator
shows that any state $\psi(Q)$ has an expansion of the form
\begin{equation}
 \psi(Q)=
 \frac{1}{\sqrt{2}}\left(\int_{-\infty}^{\infty}c_{\lambda+}\psi_{\lambda+}(Q)
   {\rm 
   d}\lambda+ \int_{-\infty}^{\infty} c_{\lambda-}\psi_{\lambda-}(Q) {\rm
   d}\lambda\right) 
\end{equation}
in terms of eigenstates of $\hat{H}_{\delta}'$, for some $c_{\lambda\pm}$
normalized such that $\int_{-\infty}^{\infty}|c_{\lambda\pm}|^2{\rm
  d}\lambda=1$. It evolves according to
\begin{equation}
  \psi(Q,\phi)=\frac{1}{\sqrt{2}}\left(
  \int_{-\infty}^{\infty}c_{\lambda+}\exp(-i\lambda\phi/\hbar)\psi_{\lambda+}(Q)
  {\rm 
    d}\lambda+ \int_{-\infty}^{\infty}
  c_{\lambda-}\exp(-i\lambda\phi/\hbar)\psi_{\lambda-}(Q){\rm d}\lambda\right) \,.
\end{equation}

The actual dynamics in our models of interest is generated by the Hamiltonian
$\hat{H}_{\delta}=|\hat{H}_{\delta}'|$. This operator has the same eigenstates
$\psi_{\lambda\pm}(Q)$, but with eigenvalues $|\lambda|$. Its spectrum is
therefore four-fold degenerate and positive. Dynamical solutions in these
models are given by
\begin{equation}
  \psi(Q,\phi)= \frac{1}{\sqrt{2}}\left(
  \int_{-\infty}^{\infty}c_{\lambda+}\exp(-i|\lambda| \phi/\hbar)\psi_{\lambda+}(Q)
  {\rm 
    d}\lambda+ \int_{-\infty}^{\infty}
  c_{\lambda-}\exp(-i|\lambda| \phi/\hbar)\psi_{\lambda-}(Q){\rm
    d}\lambda\right) \,. 
\end{equation}
The decomposition
$\psi(Q,\phi)=\frac{1}{\sqrt{2}}
\left(N_-\psi_-(Q,\phi)+N_+\psi_+(Q,\phi)\right)$
with
\begin{equation}
  \psi_-(Q,\phi)= \frac{1}{N_-}\left(
  \int_{-\infty}^{0}c_{\lambda+}\exp(i\lambda \phi/\hbar)\psi_{\lambda+}(Q)
  {\rm 
    d}\lambda+ \int_{-\infty}^{0}
  c_{\lambda-}\exp(i\lambda\phi /\hbar)\psi_{\lambda-}(Q){\rm d}\lambda\right) 
\end{equation}
and
\begin{equation}
  \psi_+(Q,\phi)= \frac{1}{N_+}\left(
  \int_0^{\infty}c_{\lambda+}\exp(-i\lambda\phi/\hbar)\psi_{\lambda+}(Q)
  {\rm 
    d}\lambda+ \int_0^{\infty}
  c_{\lambda-}\exp(-i\lambda\phi/\hbar)\psi_{\lambda-}(Q){\rm d}\lambda \right)\,,
\end{equation}
where $N_-^2=\int_{-\infty}^0(|c_{\lambda+}|^2+|c_{\lambda-}|^2){\rm d}\lambda$
and $N_+^2=\int_0^{\infty}(|c_{\lambda+}|^2+|c_{\lambda-}|^2){\rm d}\lambda$
such that $N_-^2+N_+^2=2$,
demonstrates the claim about solutions made in the introduction.

The decomposition into positive-energy solutions $\psi_+$ and negative-energy
solutions $\psi_-$ simply rewrites generic wave functions and does not
restrict their fluctuations of $Q$ or $P$. However, it is sometimes preferred
\cite{APS} (although not required \cite{GenRepIn}) to discard negative-energy
solutions and consider only positive-energy solutions $\psi_+$ (or vice verse,
but no superpositions of solutions with opposite signs of the energy).  A
question of interest in quantum cosmology is whether this restriction in any
way limits the possible magnitude of fluctuations of $Q$ or $P$, which would
then have consequences for bouncing or non-bouncing behavior according to
\cite{NonBouncing}. Using the spectral properties derived in the preceding
section, we now show that this is not the case.

In particular, for potential non-bouncing behavior, we are interested in
solutions with small $\langle\hat{H}_{\delta}\rangle$, such that 
\begin{equation} \label{HDelta} 
\delta^2\langle\hat{H}_{\delta}\rangle^2+
  \delta^2(\Delta H_{\delta})^2\leq (\Delta Q)^2- (\Delta{\rm
    Re}J_{\delta})^2\,.
\end{equation}
If $\langle\hat{H}_{\delta}\rangle$ is small, given the positivity of the
spectrum of $\hat{H}_{\delta}$, the range of possible values of $\Delta
H_{\delta}$ seems to be limited because the state in the
$\lambda$-representation can spread out only to one side of
$\langle\hat{H}_{\delta}\rangle$. However, the twofold degeneracy of the
spectrum of $\hat{H}_{\delta}'$, of the specific form derived in the preceding
section, in particular in Lemma~\ref{l:Symm}, shows that there is no such
limitation for fluctuations $\Delta Q$ even if $\langle\hat{Q}\rangle$ is
required to be small: In order to construct a state, supported only on the
positive part of the spectrum of $\hat{H}_{\delta}'$, such that it has a small
expectation value and large fluctuations of $\hat{Q}$, we choose some
$c_{\lambda}$ such that $\int_0^{\infty}|c_{\lambda}|^2{\rm d}\lambda=1$, and
define $\psi_{c+}(Q)=\int_0^{\infty} c_{\lambda}\psi_{\lambda+}(Q){\rm
  d}\lambda$. This state is supported on the positive part of the spectrum of
$\hat{H}_{\delta}'$, by construction, and has a certain expectation value
$\langle\hat{Q}\rangle_{c_+}>0$ and fluctuations $\Delta_{c+} Q>0$.
Similarly, the state $\psi_{c-}(Q)=\int_0^{\infty}
c_{\lambda}\psi_{\lambda-}(Q){\rm d}\lambda$, using the transformation
(\ref{psipm}), has expectation value
$\langle\hat{Q}\rangle_{c_-}=-\langle\hat{Q}\rangle_{c_+}-\hbar\delta<0$ and
fluctuations $\Delta_{c-} Q=\Delta_{c+} Q>0$.  The state
\begin{equation} \label{psic}
 \psi_{c}= \frac{1}{\sqrt{2}}\int_0^{\infty}
 c_{\lambda}\left(\alpha\psi_{\lambda+}(Q)+
 \sqrt{2-\alpha^2}\psi_{\lambda-}(Q)\right){\rm
   d}\lambda \,,
\end{equation}
with some $|\alpha|\leq\sqrt{2}$,
then has expectation value
\begin{equation}
  \langle\hat{Q}\rangle=\frac{1}{2}\left(
  \alpha^2\langle\hat{Q}\rangle_{c+}+
 (2-\alpha^2)\langle\hat{Q}\rangle_{c-}\right)=
(\alpha^2-1)\langle\hat{Q}\rangle_{c+}-
\frac{2-\alpha^2}{2}\hbar\delta 
\end{equation}
and fluctuations given by
\begin{eqnarray}
 (\Delta Q)^2&=&
 \frac{1}{2}\left(\alpha^2\langle\hat{Q}^2\rangle_{c+}+(2-\alpha^2)
   \langle\hat{Q}^2\rangle_{c-} \right)- \langle\hat{Q}\rangle^2\nonumber\\
&=&  (\Delta_{c+}Q)^2+ \frac{1}{2}\left(\alpha^2\langle\hat{Q}\rangle_{c+}^2+
(2-\alpha^2)\langle\hat{Q}\rangle_{c-}^2\right)\nonumber\\
&& - (\alpha^2-1)^2
\langle\hat{Q}\rangle_{c+}^2+
(2-\alpha^2)(\alpha^2-1)\hbar\delta\langle\hat{Q}\rangle_{c+}-
\frac{(2-\alpha^2)^2}{4} \hbar^2\delta^2\nonumber\\
&=& (\Delta_{c+}Q)^2+ \alpha^2(2-\alpha^2)
\left(\langle\hat{Q}\rangle_{c+}+\frac{1}{2}\hbar\delta\right)^2 \,. 
\end{eqnarray}
For $\alpha\not=1$, the result can also be written as
\begin{equation}
 (\Delta Q)^2= (\Delta_{c+}Q)^2+ \frac{\alpha^2(2-\alpha^2)}{(\alpha^2-1)^2}
 \left(\langle\hat{Q}\rangle+\frac{1}{2}\hbar\delta\right)^2
\end{equation}
using
\begin{equation}
 (\alpha^2-1)
 \left(\langle\hat{Q}\rangle_{c+}+\frac{1}{2}\hbar\delta\right)^2=
 \left(\langle\hat{Q}\rangle+\frac{1}{2}\hbar\delta\right)^2\,.
\end{equation}
Since $\langle\hat{Q}\rangle_{c+}$ is not restricted by the positivity
condition, $\Delta Q$ is unlimited even on states with small expectation value
$\langle\hat{Q}\rangle$.

\section{Moments}

Since $H_0'$ is a function of $Q$ and $P$, $H_0'$-moments in a given state are
related to $Q$ and $P$-moments in the same state. There may therefore be
restrictions on the magnitude of $Q$ or $P$-fluctuations if a state is
required to have small $\langle\hat{H}_0'\rangle$ and small
$H_0'$-fluctuations. We will now demonstrate that $Q$ and $P$-fluctuations are
indeed restricted in such a state, but only if additional assumptions on the
$QP$-covariance are made.

\subsection{Relationships between moments}

Because $\hat{H}_0'$ is quadratic in $\hat{Q}$ and $\hat{P}$, $\Delta H_0'$ is
related to moments of up to fourth order in $Q$ and $P$. In the following
calculations, we will be using the notation of \cite{Counting}, as in
\begin{defi}
  Given a set of operators $\hat{A}_i$, $i=1,\ldots,n$, and integers
  $k_1,\ldots,k_n\geq 0$ such that $\sum_ik_1\geq 2$, the {\em moments} of a
  state are
\begin{equation}
 \Delta(A_1^{k_1}A_2^{k_2}\cdots A_n^{k_n})=
 \left\langle(\Delta\hat{A}_1)^{k_1}(\Delta\hat{A}_2)^{k_2}\cdots
   (\Delta\hat{A}_n)^{k_n}\right\rangle_{\rm symm} \,,
\end{equation}
where $\Delta\hat{A}_i=\hat{A}_i-\langle\hat{A}_i\rangle$, all expectation
values are taken in the given state, and the subscript ``symm'' indicates that
all products of operators are taken in totally symmetric (or Weyl) ordering:
\begin{equation}
 \langle \hat{O}_1\cdots\hat{O}_n\rangle_{\rm symm}=
 \frac{1}{n!} \sum_{\sigma\in S_n}
 \left\langle\hat{O}_{\sigma(1)} \cdots\hat{O}_{\sigma(n)}\right\rangle\,.
\end{equation}
\end{defi}

The following reordering relations will be useful:
\begin{lemma}
 For two operators $\hat{Q}$ and $\hat{P}$ such that
 $[\hat{Q},\hat{P}]=i\hbar$,
\begin{eqnarray}
 \left\langle (\Delta\hat{Q})^2\Delta\hat{P}+
   2\Delta\hat{Q}\Delta\hat{P}\Delta\hat{Q}+
   \Delta\hat{P}(\Delta\hat{Q})^2\right\rangle&=& 4\Delta(Q^2P) \label{Delta3}\\
\left\langle\left(\Delta\hat{Q}\Delta\hat{P}+
\Delta\hat{P}\Delta\hat{Q}\right)^2\right\rangle &=&
4\Delta(Q^2P^2)+\hbar^2\,. \label{Delta4}
\end{eqnarray}
\end{lemma}
\begin{proof}
 Starting with the left-hand side of (\ref{Delta3}), we write
\begin{eqnarray*}
 2\Delta\hat{Q}\Delta\hat{P}\Delta\hat{Q} &=&
 \frac{4}{3}\Delta\hat{Q}\Delta\hat{P}\Delta\hat{Q}\\
&&+ \frac{1}{3}
 \left((\Delta\hat{Q})^2\Delta\hat{P}-
   \Delta\hat{Q}[\Delta\hat{Q},\Delta\hat{P}]\right)+ \frac{1}{3} 
 \left(\Delta\hat{P}(\Delta\hat{Q})^2+
   [\Delta\hat{Q},\Delta\hat{P}]\Delta\hat{Q}\right)
\end{eqnarray*}
such that
\begin{eqnarray*}
 \left\langle (\Delta\hat{Q})^2\Delta\hat{P}+
   2\Delta\hat{Q}\Delta\hat{P}\Delta\hat{Q}+
   \Delta\hat{P}(\Delta\hat{Q})^2\right\rangle &=& \frac{4}{3} \left\langle
   (\Delta\hat{Q})^2\Delta\hat{P}+ 
   \Delta\hat{Q}\Delta\hat{P}\Delta\hat{Q}+
   \Delta\hat{P}(\Delta\hat{Q})^2\right\rangle\\
&=& 4\Delta(Q^2P)
\end{eqnarray*}
proves (\ref{Delta3}).

On the left-hand side of (\ref{Delta4}), we 
write
\begin{eqnarray*}
&& \left\langle\left(\Delta\hat{Q}\Delta\hat{P}+
\Delta\hat{P}\Delta\hat{Q}\right)^2\right\rangle\\
 &=& \left\langle
\Delta\hat{Q}\Delta\hat{P}\Delta\hat{Q}\Delta\hat{P}+
\Delta\hat{Q}(\Delta\hat{P})^2\Delta\hat{Q}+
\Delta\hat{P}(\Delta\hat{Q})^2\Delta\hat{P}+
\Delta\hat{P}\Delta\hat{Q}\Delta\hat{P}\Delta\hat{Q}\right\rangle\\
&=& \frac{2}{3}\left\langle
  \Delta\hat{Q}\Delta\hat{P}\Delta\hat{Q}\Delta\hat{P}
  +(\Delta\hat{Q})^2(\Delta\hat{P})^2+ 
\Delta\hat{Q}(\Delta\hat{P})^2\Delta\hat{Q}\right.\\
&&\qquad\left.+
\Delta\hat{P}(\Delta\hat{Q})^2\Delta\hat{P}+
(\Delta\hat{P})^2(\Delta\hat{Q})^2+ 
\Delta\hat{P}\Delta\hat{Q}\Delta\hat{P}\Delta\hat{Q}\right\rangle\\
&&+\frac{1}{3} \left(\Delta\hat{Q}[\Delta\hat{P},\Delta\hat{Q}]\Delta\hat{P}+
  \Delta\hat{Q}[(\Delta\hat{P})^2,\Delta\hat{Q}]+
  \Delta\hat{P}[(\Delta\hat{Q})^2,\Delta\hat{P}]+
  \Delta\hat{P}[\Delta\hat{Q},\Delta\hat{P}]\Delta\hat{Q}\right)
\end{eqnarray*}
using 
\begin{eqnarray*}
 \Delta\hat{Q}\Delta\hat{P}\Delta\hat{Q}\Delta\hat{P} &=& \frac{2}{3}
 \Delta\hat{Q}\Delta\hat{P}\Delta\hat{Q}\Delta\hat{P}+
 \frac{1}{3}\left((\Delta\hat{Q})^2(\Delta\hat{P})^2+ 
 \Delta\hat{Q}[\Delta\hat{P},\Delta\hat{Q}]\Delta\hat{P}\right)\\
\Delta\hat{Q}(\Delta\hat{P})^2\Delta\hat{Q}&=& \frac{2}{3}
\Delta\hat{Q}(\Delta\hat{P})^2\Delta\hat{Q}+
\frac{1}{3}\left((\Delta\hat{Q})^2(\Delta\hat{P})^2+
  \Delta\hat{Q}[(\Delta\hat{P})^2,\Delta\hat{Q}]\right)\\
\Delta\hat{P}(\Delta\hat{Q})^2\Delta\hat{P} &=& \frac{2}{3}
\Delta\hat{P}(\Delta\hat{Q})^2\Delta\hat{P}+
\frac{1}{3}\left((\Delta\hat{P})^2(\Delta\hat{Q})^2+
  \Delta\hat{P}[(\Delta\hat{Q})^2,\Delta\hat{P}]\right)\\
\Delta\hat{P}\Delta\hat{Q}\Delta\hat{P}\Delta\hat{Q}&=& \frac{2}{3}
\Delta\hat{P}\Delta\hat{Q}\Delta\hat{P}\Delta\hat{Q}+
\frac{1}{3}\left((\Delta\hat{P})^2(\Delta\hat{Q})^2+ 
\Delta\hat{P}[\Delta\hat{Q},\Delta\hat{P}]\Delta\hat{Q}\right)\,.
\end{eqnarray*}
Evaluating the commutators and observing
\begin{eqnarray*}
 \Delta(Q^2P^2)&=& \frac{1}{6}\left\langle
  \Delta\hat{Q}\Delta\hat{P}\Delta\hat{Q}\Delta\hat{P}
  +(\Delta\hat{Q})^2(\Delta\hat{P})^2+ 
\Delta\hat{Q}(\Delta\hat{P})^2\Delta\hat{Q}\right.\\
&&\qquad\left.+
\Delta\hat{P}(\Delta\hat{Q})^2\Delta\hat{P}+
(\Delta\hat{P})^2(\Delta\hat{Q})^2+ 
\Delta\hat{P}\Delta\hat{Q}\Delta\hat{P}\Delta\hat{Q}\right\rangle
\end{eqnarray*}
 we obtain (\ref{Delta4}).
\end{proof}

\begin{prop} \label{Prop}
  If a state is such that it has a vanishing covariance
  $\Delta(QP)=\frac{1}{2}\langle\hat{Q}\hat{P}+\hat{P}\hat{Q}\rangle-
  \langle\hat{Q}\rangle\langle\hat{P}\rangle$ and zero skewness (third-order
  moments), then the relative fluctuations of $\hat{Q}$ and $\hat{P}$ are
  bounded from above by the relative fluctuation of
  $\hat{H}=\frac{1}{2}(\hat{Q}\hat{P}+\hat{P}\hat{Q})$:
\begin{equation} \label{InEq}
 \frac{(\Delta Q)^2}{\langle\hat{Q}\rangle^2}+ \frac{(\Delta
   P)^2}{\langle\hat{P}\rangle^2}< \frac{(\Delta
   H_0')^2}{\langle\hat{H}_0'\rangle^2}\,.
\end{equation}
\end{prop}

\begin{proof}
  Writing operators as $\hat{A}=\Delta\hat{A}+\langle\hat{A}\rangle$ in
  $\hat{H}_0'=\frac{1}{2}(\hat{Q}\hat{P}+\hat{P}\hat{Q})$, we obtain
\begin{eqnarray}
 \langle\hat{H}_0'\rangle&=& \frac{1}{2} \left\langle
 (\Delta\hat{Q}+\langle\hat{Q}\rangle) (\Delta\hat{P}+\langle\hat{P}\rangle)+
 (\Delta\hat{P}+\langle\hat{P}\rangle)
 (\Delta\hat{Q}+\langle\hat{Q}\rangle)\right\rangle\nonumber\\
&=& \frac{1}{2} \langle
\Delta\hat{Q}\Delta\hat{P}+\Delta\hat{P}\Delta\hat{Q}\rangle +
\langle\hat{Q}\rangle \langle\hat{P}\rangle= \Delta(QP)+
\langle\hat{Q}\rangle \langle\hat{P}\rangle\,.
\end{eqnarray}
(Note that $\langle\Delta\hat{A}\rangle=0$ for any $\hat{A}$.)  

The derivation
of the fluctuation $\Delta(H_0'{}^2)$ requires a longer calculation: We expand
\begin{eqnarray}
\Delta(H_0'{}^2) &=& \langle\hat{H}_0'{}^2\rangle-\langle\hat{H}_0'\rangle^2\\
&=& \frac{1}{4}\left\langle\left(
    (\Delta\hat{Q}+\langle\hat{Q}\rangle)
    (\Delta\hat{P}+\langle\hat{P}\rangle)+
    (\Delta\hat{P}+\langle\hat{P}\rangle)
    (\Delta\hat{Q}+\langle\hat{Q}\rangle)\right)^2\right\rangle\\
 &&-
\frac{1}{4}\left\langle (\Delta\hat{Q}+\langle\hat{Q}\rangle)
    (\Delta\hat{P}+\langle\hat{P}\rangle)+
    (\Delta\hat{P}+\langle\hat{P}\rangle)
    (\Delta\hat{Q}+\langle\hat{Q}\rangle)\right\rangle^2\\
&=& \frac{1}{4}
\left\langle\left(\Delta\hat{Q}\Delta\hat{P}+\Delta\hat{P}\Delta\hat{Q}+
    2\langle\hat{P}\rangle\Delta\hat{Q}+ 2\langle\hat{Q}\rangle\Delta\hat{P}+
    2\langle\hat{Q}\rangle\langle\hat{P}\rangle\right)^2\right\rangle\\
&&-
\frac{1}{4}\left(\left\langle\Delta\hat{Q}\Delta\hat{P}+
\Delta\hat{P}\Delta\hat{Q}\right\rangle+
2\langle\hat{Q}\rangle\langle\hat{P}\rangle\right)^2\\
&=& \frac{1}{4}
\left\langle\left(\Delta\hat{Q}\Delta\hat{P}+ 
\Delta\hat{P}\Delta\hat{Q}\right)^2\right\rangle \label{H1}\\
&&+
\frac{1}{2}\langle\hat{Q}\rangle 
\left\langle\Delta\hat{P}\left(\Delta\hat{Q}\Delta\hat{P}+
    \Delta\hat{P}\Delta\hat{Q}\right)+ \left(\Delta\hat{Q}\Delta\hat{P}+
    \Delta\hat{P}\Delta\hat{Q}\right) \Delta\hat{P}\right\rangle\label{H2}\\
&&+
\frac{1}{2}\langle\hat{P}\rangle 
\left\langle\Delta\hat{Q}\left(\Delta\hat{Q}\Delta\hat{P}+
    \Delta\hat{P}\Delta\hat{Q}\right)+ \left(\Delta\hat{Q}\Delta\hat{P}+
    \Delta\hat{P}\Delta\hat{Q}\right) \Delta\hat{Q}\right\rangle\label{H3}\\
&&+\langle\hat{Q}\rangle \langle\hat{P}\rangle
\left\langle\Delta\hat{Q}\Delta\hat{P}+ 
    \Delta\hat{P}\Delta\hat{Q}\right\rangle
  +\langle\hat{Q}\rangle^2\langle(\Delta\hat{P})^2\rangle+ 
  \langle\hat{P}\rangle \langle(\Delta\hat{Q})^2\rangle\\
&&- \frac{1}{4}
  \langle\Delta\hat{Q}\Delta\hat{P}+  \Delta\hat{P}\Delta\hat{Q}\rangle^2\,.
\end{eqnarray}
Using (\ref{Delta4}) in line (\ref{H1}), (\ref{Delta3}) in line (\ref{H3}) and
an analogous result in line (\ref{H2}), we obtain
\begin{eqnarray}
  \Delta(H_0'{}^2) &=& \langle\hat{Q}\rangle^2\Delta(P^2)+
  \langle\hat{P}\rangle^2\Delta(Q^2)+ 
  2\langle\hat{Q}\rangle\langle\hat{P}\rangle\Delta(QP)\nonumber\\
  &&+
  2\langle\hat{P}\rangle\Delta(Q^2P)+ 2\langle\hat{Q}\rangle\Delta(QP^2)+
  \Delta(Q^2P^2)+\frac{1}{4}\hbar^2-   \Delta(QP)^2\,.
\end{eqnarray}
If $\Delta(QP)=0$ and $\Delta(Q^2P)=0=\Delta(QP^2)$, we obtain
\begin{equation}
 \frac{\Delta(H_0'{}^2)}{\langle\hat{H}_0'\rangle^2}=
 \frac{\Delta(Q^2)}{\langle\hat{Q}\rangle^2}+
 \frac{\Delta(P^2)}{\langle\hat{P}\rangle^2}+ \frac{1}{4}
 \frac{\hbar^2+4\Delta(Q^2P^2)}{\langle\hat{Q}\rangle^2\langle\hat{P}\rangle^2}> 
\frac{\Delta(Q^2)}{\langle\hat{Q}\rangle^2}+
 \frac{\Delta(P^2)}{\langle\hat{P}\rangle^2}\,.
\end{equation}
\end{proof}

This result shows that a state with small relative $H$-fluctuations but large
relative $Q$-fluctuations must have non-zero covariance or skewness.

\subsection{Example}

As shown in \cite{FluctEn}, the right-hand side of (\ref{HDelta}) is strictly
negative for a Gaussian state in $Q$. This inequality then cannot be
fulfilled. The same paper showed that the right-hand side of (\ref{HDelta}) is
approximately zero for a state given by
\begin{equation} \label{Qlog}
 \psi(Q)=\sqrt{\sqrt{\frac{2}{\pi}}\frac{\sigma}{\hbar Q}} \exp\left(-\sigma^2
   (\log(Q/\bar{Q}))^2/\hbar^2+ i(\bar{\lambda}/\hbar)
   \log(Q/\bar{Q})\right)
\end{equation}
if $Q>0$ and $\psi(Q)=0$ otherwise,
with constants $\bar{Q}>0$, $\sigma>0$ and $\bar{\lambda}$. We now demonstrate
that such a state can be approximated by a state supported only on the
positive part of the spectrum of $\hat{H}_0'$, which then provides an example
of how the restriction given by Proposition~\ref{Prop} can be overcome by
states with non-zero covariance.

Let us choose a Gaussian
\begin{equation} \label{cGaussian}
 c_{\lambda}=\frac{N}{(2\pi)^{1/4}\sqrt{\sigma}}
 \exp\left(-\frac{(\lambda-\bar{\lambda})^2}{4\sigma^2}
+\frac{i\bar{p}\lambda}{\hbar}\right)
\end{equation}
for $\lambda>0$ and $c_{\lambda}=0$ otherwise, where
\begin{equation}
 N^2=\frac{2}{1+{\rm
       erf}\left(\bar{\lambda}/(\sqrt{2}\sigma)\right)} 
\end{equation}
normalizes $c_{\lambda}$ restricted to positive $\lambda$ and is close to
$N^2\approx 1$ for $\bar{\lambda}\gg\sigma$, or $\Delta
H_0'/\langle\hat{H}_0'\rangle\ll 1$.  Using the definition (\ref{psic}) with
$\alpha=\sqrt{2}$, we consider the state $\psi_{c+}(Q) = \int_0^{\infty}
c_{\lambda}\psi_{\lambda+}(Q){\rm d}\lambda$. The integral can be approximated
by extending the integration over positive $\lambda$ to all real $\lambda$,
which is valid provided $c_{\lambda}$ is negligible for $\lambda<0$. Given
(\ref{cGaussian}), the approximation can be used if the $\lambda$-variance
$\sigma$ is much less than the $\lambda$-expectation value,
$\sigma\ll\bar{\lambda}$. The same condition allows us to approximate
$N\approx 1$, and we obtain
\begin{eqnarray} \label{psi}
 \psi_{c+}(Q) &=& \int_0^{\infty}
 c_{\lambda}\psi_{\lambda+}(Q){\rm d}\lambda 
 \approx  \int_{-\infty}^{\infty} c_{\lambda}
 \psi_{\lambda+}(Q) {\rm d}\lambda\nonumber\\
&\approx& \frac{1}{(2\pi)^{3/4} \sqrt{\sigma\hbar Q}} \int_{-\infty}^{\infty}
  \exp(- (\lambda-\bar{\lambda})^2/4\sigma^2+ i\lambda(\bar{p}+\log
  Q)/\hbar) {\rm d}\lambda\\
&=& \sqrt{\sqrt{\frac{2}{\pi}}\frac{\sigma}{\hbar Q}}
\exp\left(-(\sigma^2/\hbar^2) 
  (\bar{p}+\log Q)^2+ i(\bar{\lambda}/\hbar) (\bar{p}+\log Q)\right)
\end{eqnarray}
for $Q>0$.  Defining $\bar{Q}=\exp(-\bar{p})$, the result equals (\ref{Qlog}).

The resulting state (\ref{psi}) shows that the $\log|Q|$-variance is given by
$\Delta\log |Q|=\hbar/(2\sigma)$, while the $\log|Q|$-expectation value is
$\langle\log|\hat{Q}|\rangle=-\bar{p}$. We can therefore maintain the
condition $\bar{\lambda}\gg\sigma$, or $\Delta
H_0'/\langle\hat{H}_0'\rangle\ll 1$, for the approximation in (\ref{psi}) to be
valid, and choose a small $\langle\hat{Q}\rangle$ with large $\Delta Q$.

According to (\ref{InEq}), this state must have non-zero covariance or
skewness. We can easily confirm the former property by computing
\begin{eqnarray}
 \langle\hat{Q}\rangle &=& \sqrt{\frac{2}{\pi}} \frac{\sigma}{\hbar}
 \int_0^{\infty} \exp\left(-2(\sigma^2/\hbar^2)
   (\log(Q/\bar{Q}))^2\right){\rm d}Q=I_1
=\bar{Q}\exp(\hbar^2/8\sigma^2)\\ 
 \langle\hat{P}\rangle&=& \sqrt{\frac{2}{\pi}}
 \frac{\sigma\bar{\lambda}}{\hbar} 
 \int_0^{\infty} \frac{1}{Q^2} \exp\left(-2(\sigma^2/\hbar^2)
   (\log(Q/\bar{Q}))^2\right){\rm d}Q=\bar{\lambda}I_{-1}
=\frac{\bar{\lambda}}{\bar{Q}} \exp(\hbar^2/8\sigma^2)
\end{eqnarray}
and
\begin{equation}
  \frac{1}{2}\langle\hat{Q}\hat{P}+\hat{P}\hat{Q}\rangle = {\rm
    Re}\langle\hat{Q}\hat{P}\rangle=  \bar{\lambda}I_0=
  \bar{\lambda}\,.
\end{equation}
where we have used the integrals
\begin{equation}
 I_a = \sqrt{\frac{2}{\pi}} \frac{\sigma}{\hbar} \int_{-\infty}^{\infty}
 e^{az} \exp(-2\sigma^2(z-\log\bar{Q})^2/\hbar^2){\rm d}z= \bar{Q}^a
 \exp(a^2\hbar^2/8\sigma^2)
\end{equation}
for real $a$.  Therefore,
\begin{equation}
 \Delta(QP)=\bar{\lambda}\left(1-\exp(\hbar^2/4\sigma^2)\right)<0
\end{equation}
is non-zero, with $|\Delta(QP)|$ large for $\sigma\ll\hbar$, such that $\Delta
Q$ can be large.

\section*{Acknowledgements}

This work was supported in part by NSF grant PHY-1607414.


\end{document}